\documentclass[letterpaper]{article} % DO NOT CHANGE THIS

\ifdefined\aaaianonymous
    \usepackage[submission]{aaai2027}
\else
    \usepackage{aaai2027}
\fi
\nocopyright % suppress the submission copyright footnote (anonymity is unaffected)

\usepackage{times}
\usepackage{helvet}
\usepackage{courier}
\usepackage[hyphens]{url}
\usepackage{graphicx}
\usepackage{natbib}
\usepackage{caption}
\usepackage{amsmath}
\usepackage{amssymb}
\usepackage{amsthm}
\usepackage{booktabs}
\usepackage{tikz}
\usetikzlibrary{patterns,positioning,arrows.meta,fit,backgrounds}
\usepackage{colortbl}
\usepackage{pgfplots}\pgfplotsset{compat=1.18}

\theoremstyle{plain}
\newtheorem{theorem}{Theorem}
\newtheorem{proposition}{Proposition}
\newtheorem{lemma}{Lemma}

\theoremstyle{definition}
\newtheorem{definition}{Definition}
\newtheorem{example}{Example}

\newcommand{\Up}[1]{{\uparrow}#1}
\newcommand{\Kstar}{K^{\!*}}
\newcommand{\Dstar}{D^{\!*}}

\newcommand{\AXp}{\textnormal{\textsc{AXp}}}

\let\oldtextsc\textsc
\renewcommand{\textsc}[1]{\textnormal{\oldtextsc{#1}}}

\title{The Complexity of Auditing Disclosure-Robust Defeasible Explanations}

\author{Haoyang Li}
\affiliations{University of Sydney\\ \texttt{ygglc123@gmail.com}}

\begin{document}

\maketitle

\begin{abstract}
A formal explanation certifies a prediction with a subset-minimal sufficient reason. Under \emph{incremental disclosure}, however, evidence arrives field by field, and a normally sufficient reason can be overturned by later information. We study the smallest reason \emph{core that remains sufficient under all admissible later disclosures}; its size is the \emph{robustness radius}. We compile a defeasible classifier into an explicit \emph{boundary atlas} of entry anchors and exit defeaters, and chart the complexity of auditing it (all statements are in the atlas size). Prediction and standing anchors are read by polynomial-time scans of the atlas, without iterative fixpoint computation; a reason's defeater frontier is obtained by scanning and subset-minimizing the defeaters above it. But \emph{verifying} that a reason core is robust is coNP-complete, and \emph{deciding} whether a robust core of size $\le\theta$ exists is $\Sigma_2^p$-complete---a four-cell P\,/\,coNP-c\,/\,NP-c\,/\,$\Sigma_2^p$-c landscape, with the accepted ($A(t){=}1$) case reaching the second level of the polynomial hierarchy. The decision version of minimal certified disclosure is NP-complete; its optimization version is fixed-parameter tractable in the number of excluded worlds without defeaters, with the general-defeater case open. On exact audits of depth-limited decision trees over standard tabular datasets under a deliberately small Boolean abstraction (credit, income, medical) the governing parameters fall in a small-parameter regime (robust cores in the low single digits, $2$--$4$ across datasets and seeds), so exact robust auditing is tractable in these audited cubes; on adversarial instances built from our reductions the hardness bites, with robust cores of size $\Theta(n)$ and a greedy optimality gap that widens with $n$. To our knowledge this is the first $\Sigma_2^p$-complete audit query for disclosure-robust formal explanations.
\end{abstract}

\section{Introduction}

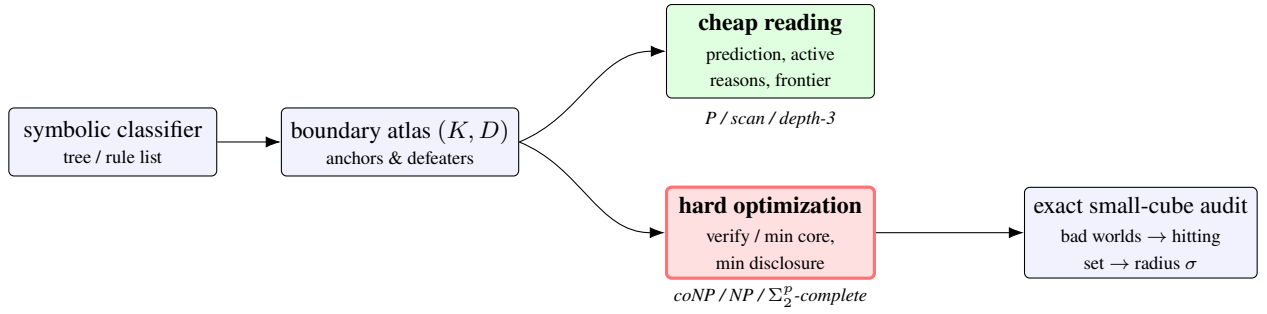
\begin{figure*}[t]
\centering
\begin{tikzpicture}[font=\footnotesize,>={Latex[length=2mm]},
  box/.style={draw,rounded corners=2pt,align=center,inner sep=3.5pt,minimum height=9mm,minimum width=2.75cm,fill=blue!5},
  cheap/.style={box,fill=green!12},
  hard/.style={box,fill=red!12,draw=red!55,very thick},
  sub/.style={font=\scriptsize\itshape,inner sep=1pt}]
  \node[box]   (clf)   at (0,0)       {symbolic classifier\\{\scriptsize tree / rule list}};
  \node[box]   (atlas) at (3.8,0)     {boundary atlas $(K,D)$\\{\scriptsize anchors \& defeaters}};
  \node[cheap] (read)  at (8.7,1.2)   {\textbf{cheap reading}\\[1pt]{\scriptsize prediction, active}\\{\scriptsize reasons, frontier}};
  \node[hard]  (audit) at (8.7,-1.2)  {\textbf{hard optimization}\\[1pt]{\scriptsize verify / min core,}\\{\scriptsize min disclosure}};
  \node[box]   (reg)   at (13.6,-1.2) {exact small-cube audit\\[1pt]{\scriptsize bad worlds $\to$ hitting}\\{\scriptsize set $\to$ radius $\sigma$}};
  \draw[->] (clf)--(atlas);
  \draw[->] (atlas.east) to[out=22,in=180] (read.west);
  \draw[->] (atlas.east) to[out=-22,in=180] (audit.west);
  \draw[->] (audit)--(reg);
  \node[sub] at (8.7,0.30)  {P / scan / depth-3};
  \node[sub] at (8.7,-2.05) {coNP / NP / $\Sigma_2^p$-complete};
\end{tikzpicture}
\caption{The audit pipeline. A symbolic classifier is compiled once into a boundary atlas; \emph{reading} it---prediction, active reasons, the defeater frontier---is cheap, while \emph{optimizing} the audit---verifying a robust core, minimizing it, minimal certified disclosure---climbs to the second level of the polynomial hierarchy. On the deliberately small Boolean cubes we audit, the governing parameter stays small, so exact robust-core search is tractable in that regime.}
\label{fig:pipeline}
\end{figure*}

A formal explanation answers a sharp question---\emph{why} did the classifier decide this?---by naming a subset-minimal set of feature values that force the prediction, whatever the rest \cite{ignatiev2019abduction,marquessilva2023logic}. But it answers at a \emph{frozen} instance over a \emph{fixed} feature space, and audits are rarely frozen. A compliance review uncovers evidence one field at a time, a record is completed step by step, and the governing rules are \emph{defeasible}: a reason that suffices today can be defeated by tomorrow's disclosure and then reinstated the day after. A loan approved on ``stable income, low debt'' is overturned the moment a recent default surfaces, and restored once a guarantor is verified. Against this moving target a single sufficient reason is brittle---it neither survives the next disclosure nor names what would break it.

What an auditor needs is not a sufficient reason but a sufficient \emph{core}: a set of early disclosures whose verdict no \emph{further} admissible disclosure can overturn, and ideally the smallest one, whose size we call the \emph{robustness radius}. We ask how hard that object is to find, to verify, and to certify minimally---and the answer is sharp. Reading a compiled model is cheap, but pinning down its minimal robust core sits at the \emph{second level} of the polynomial hierarchy. The worst case is not the common case, though: on real audited models the parameter that governs exact auditing stays small, so robust auditing is second-level-hard in principle yet tractable in the small-parameter regime these audits occupy. The barrier and the parameter, together, are the message.

To make this precise we compile the classifier into a \emph{boundary atlas} of \emph{entry anchors} (that turn the verdict on) and \emph{exit defeaters} (that turn it off)---the local lower-cover faces capturing the $1{\to}0{\to}1$ life of a defeasible reason. On this compiled object (Figure~\ref{fig:pipeline}) the audit queries split into cheap \emph{reading} and hard \emph{optimization}, the latter a ladder P / coNP-c / NP-c / $\Sigma_2^p$-c topped by the minimal robust core.

Each piece of this is, on its own, familiar. The complexity of explanation \emph{queries} has been mapped across model classes \cite{waldchen2021complexity,barcelo2020interpretability}; knowledge compilation makes online queries cheap once a normal form is precomputed \cite{audemard2020tractable,darwiche2002kcmap}; and sufficient-reason and shortest-implicant problems are known to reach $\Sigma_2^p$ \cite{umans2001shortest}. \emph{What is missing is the defeasible, incremental-disclosure axis}: every one of these results fixes the instance and the feature space, so none asks whether a reason \emph{survives what is disclosed next}, or what certifying that survival costs. We supply exactly this, and find its central query---the minimal robust core---to be $\Sigma_2^p$-complete: to our knowledge the first time disclosure-robustness is placed at the second level of the polynomial hierarchy.

\paragraph{Contributions.}
\begin{itemize}
\itemsep0.15em
\item \emph{A complete audit landscape.} Verifying a robust core is in \textbf{P} (rejected side) and \textbf{coNP-complete} (accepted side); deciding whether a robust core of size $\le\theta$ exists is \textbf{NP-complete} (rejected) and \textbf{$\Sigma_2^p$-complete} (accepted). The accepted-side $\Sigma_2^p$-completeness is the technical heart (Theorems~\ref{thm:landscape}--\ref{thm:sigma2}): the reduction is no shortest-implicant re-encoding---it must keep the audited record accepted, deny a small core the trivial standing anchor, and realize the universal quantifier only through admissible later disclosures.
\item \emph{The hardness is in the optimizing, not the reading.} Prediction, active reasons, and a reason's defeater frontier are cheap atlas queries (constant-depth or a single polynomial scan), so the hardness lives entirely in \emph{minimizing} robust cores and disclosures (Theorem~\ref{thm:cheap}). Minimal certified disclosure has an NP-complete decision version, fixed-parameter tractable without defeaters (Theorem~\ref{thm:disc}).
\item \emph{These are recognizable explanations.} Atlas entries are the positive-literal \AXp{}s in the monotone fragment, and in general a robust core is the minimum sufficient reason stable under disclosure (Proposition~\ref{prop:bridge})---so the hard queries are explanation queries, not abstract combinatorics.
\item \emph{Hard in theory, small in practice.} On exact small-Boolean tabular audits the governing parameter stays in the low single digits over $k=6,\dots,12$, while constructed adversarial instances realize the worst-case blow-up.
\end{itemize}

\section{Preliminaries}

\paragraph{Disclosure model.}
Let $H=\{1,\dots,n\}$ be Boolean \emph{fields}. A \emph{record} $t\subseteq H$ is the set of disclosed (present) fields; $2^H$ is ordered by $\subseteq$ and $\Up k=\{t:k\subseteq t\}$. A classifier is a predicate $A:2^H\to\{0,1\}$ with the null-default axiom $A(\emptyset)=0$. $A$ is \emph{defeasible}: it need not be monotone---disclosing a field can flip the prediction from $1$ to $0$ and a further disclosure can flip it back. We binarize numeric features, the setting of formal XAI for trees and rule lists \cite{izza2022trees}. This is a \emph{one-sided} evidence-atom abstraction: an atom $h$ means a particular piece of evidence \emph{has been disclosed}, and the order is disclosure (set inclusion), not feature value. Signed value-literals such as $x_i{=}0$ and $x_i{=}1$ can be modelled as separate atoms; for such encodings the audit universe $t$ is taken to be consistency-respecting for the case (so $[X,t]$ never contains a record with both $x_i{=}0$ and $x_i{=}1$), or else consistency is enforced by guards compiled into $A$. Our exact experiments use the simpler one-bit-per-feature abstraction, which avoids this issue and keeps the $2^k$ cube, and hence robust-core search, exhaustive (Limitations).

\paragraph{Boundary atlas.}
An \emph{anchored pair} $(K,D)$ of finite families ($K\cap D=\emptyset$, $\emptyset\notin K$) induces
\[
A_{K,D}(t){=}1 \iff \exists k\in K\,(k\subseteq t)\wedge\nexists d\in D\,(k\subseteq d\subseteq t).
\]
We say $k$ is \emph{standing} at $t$ when $k\subseteq t$ and no $d\in D$ has $k\subseteq d\subseteq t$. The \emph{canonical} pair collects the \emph{lower-cover boundary faces} of $A$: $\Kstar=\{k\neq\emptyset:A(k){=}1,\ \exists h\in k\,A(k{\setminus}\{h\}){=}0\}$ (entries) and $\Dstar=\{d:A(d){=}0,\ \exists h\in d\,A(d{\setminus}\{h\}){=}1\}$ (exits)---these are local boundary faces, \emph{not} globally minimal accepting/rejecting records, so they capture the $1{\to}0{\to}1$ re-entries of a non-monotone $A$, and $A_{\Kstar,\Dstar}=A$ for every finite $A$ with $A(\emptyset){=}0$ (representation lemma, appendix). Any explicit $(K,D)$ is evaluated in polynomial time; the atlas is our compiled, queryable object (Figure~\ref{fig:atlas}). Throughout, a complexity statement takes an \emph{explicit} anchored pair as input, measured in $|K|+|D|$; we say \emph{canonical} only where a result needs the lower-cover atlas specifically. The atlases built in our hardness reductions are explicit anchored pairs---valid entries/exits of the predicate they induce, but not assumed canonical---so the bounds concern auditing a given compiled atlas whatever its provenance.

\paragraph{Robust cores.}
Here $t$ is the \emph{audit universe}: the set of fields admissible for eventual disclosure in this audit (the complete eventual record), \emph{not} the currently disclosed partial record. A set $X\subseteq t$ is a \emph{robust core} of $t$ if $A(s)=A(t)$ for every $s\in[X,t]$, i.e.\ the label is invariant under every disclosure admissible in the universe $t$; $X$ is the early-disclosure core and $s$ ranges over all intermediate disclosures $X\subseteq s\subseteq t$ (the maximal universe $t=H$ admits every field). When $A(t)=1$ this is exactly ``disclosing $X$ certifies the positive prediction, and no further admissible disclosure overturns it''; the \emph{robustness radius} is the \emph{size} of a smallest such core, $\sigma(t)=\min\{|X|:X\text{ a robust core of }t\}$. Equivalently, $X$ is a robust core of an accepting $t$ iff it hits the complement $t\setminus s$ of every maximal rejected $s\subseteq t$, so $\sigma(t)$ is a minimum hitting set over the bad worlds---the structure the $\Sigma_2^p$ and FPT results turn on.

\paragraph{Frontiers and audit queries.} For an anchor $k$, its \emph{global} frontier $D_k^{H}=\min_{\subseteq}\{d\in D:k\subseteq d\}$ names the counterfactual disclosures that would defeat the branch somewhere in $H$, while the \emph{universe-relative} frontier $D_k^{t}=\min_{\subseteq}\{d\in D:k\subseteq d\subseteq t\}$ is the part that actually constrains robustness over $[X,t]$---contrastive information versus a robustness bound. The audit queries then take an explicit atlas $(K,D)$ and a record $t$ with the promise $A(t){=}c$: \textsc{verify-core}$_c$ asks whether a given $X\subseteq t$ is a robust core, and \textsc{min-core}$_c$ whether $\sigma(t)\le\theta$, each in an accepted ($c{=}1$) and a rejected ($c{=}0$) version. \emph{All complexity statements are in the explicit atlas size}: compiling a succinct classifier into its full boundary atlas can be exponential, as in knowledge compilation, so the question is which queries stay cheap after compilation and which remain hard \emph{even on the compiled object}.

\begin{example}[Defeasible credit audit]\label{ex:credit}
Over fields \textsf{income\_stable}, \textsf{low\_debt}, \textsf{recent\_default}, \textsf{verified\_guarantor}, the approve rule has entry $k_1=\{\textsf{income\_stable},\textsf{low\_debt}\}$; disclosing a recent default defeats it ($d_1=k_1\cup\{\textsf{recent\_default}\}$, reject), and a verified guarantor reinstates approval ($k_2=d_1\cup\{\textsf{verified\_guarantor}\}$). A positive branch reason for approval is $k_1$; a \emph{signed} classical \AXp{} of a complete instance may additionally encode absent defeaters, but the positive branch reason alone is \emph{brittle} under future disclosure (disclosing \textsf{recent\_default} overturns it). An auditor asking ``which disclosures keep this approval certified?'' needs the disclosure-robust sufficient set together with the defeater frontier; that robust core is the object whose complexity we chart. (This is the four-field reinstatement fragment of the eight-field rule list audited in Section~\ref{sec:exp}, shown here for intuition.)
\end{example}

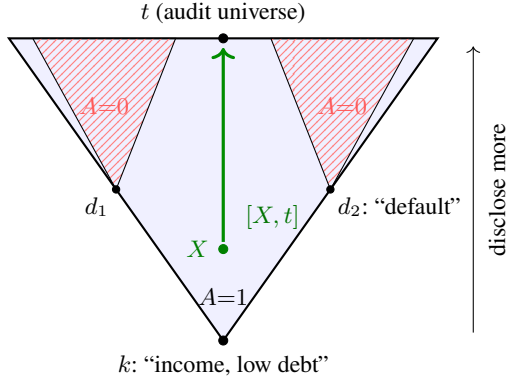
\begin{figure}[t]
\centering
\begin{tikzpicture}[scale=1.05,font=\footnotesize]
  % up-cone of entry reason k (apex k at bottom, widens upward = more disclosure)
  \fill[blue!6] (0,0) -- (-2.7,3.8) -- (2.7,3.8) -- cycle;
  \draw[thick] (0,0) -- (-2.7,3.8) -- (2.7,3.8) -- cycle;
  \filldraw (0,0) circle (1.6pt) node[below=2pt] {$k$: ``income, low debt''};
  \node at (0,0.55) {$A{=}1$};
  % defeater wedges pushed outward, leaving a clear central channel
  \fill[pattern=north east lines,pattern color=red!50] (-1.35,1.9) -- (-2.4,3.8) -- (-0.6,3.8) -- cycle;
  \fill[pattern=north east lines,pattern color=red!50] (1.35,1.9) -- (0.6,3.8) -- (2.4,3.8) -- cycle;
  \draw (-1.35,1.9) -- (-2.4,3.8) -- (-0.6,3.8) -- cycle;
  \draw (1.35,1.9) -- (0.6,3.8) -- (2.4,3.8) -- cycle;
  \filldraw (-1.35,1.9) circle (1.3pt) node[below left=-1pt] {$d_1$};
  \filldraw (1.35,1.9) circle (1.3pt) node[below right=-1pt] {$d_2$: ``default''};
  \node[red!65] at (-1.5,2.95) {$A{=}0$};
  \node[red!65] at (1.5,2.95) {$A{=}0$};
  % robust core X and its safe interval up the central channel to t
  \filldraw (0,3.8) circle (1.6pt) node[above=2pt] {$t$ (audit universe)};
  \filldraw[green!50!black] (0,1.15) circle (1.6pt) node[left=2pt] {$X$};
  \draw[->,very thick,green!55!black] (0,1.24) -- (0,3.68);
  \node[green!45!black,anchor=west] at (0.16,1.55) {$[X,t]$};
  % disclose-more axis
  \draw[->] (3.15,0.1) -- (3.15,3.7);
  \node[rotate=90] at (3.45,1.9) {disclose more};
\end{tikzpicture}
\caption{A boundary atlas. The triangle is the up-cone $\Up k$ of records extending an entry reason $k$; exit defeaters $d_1,d_2$ carve out the regions where the prediction flips to $0$ (a later ``default'' disclosure defeats the loan reason). A \emph{robust core} $X$ is an early-disclosure set whose whole interval $[X,t]$---every record from $X$ up to the audit universe $t$---keeps the verdict $A(t)$; the green path threads the defeater frontier. Minimizing such an $X$ is what makes auditing hard.}
\label{fig:atlas}
\end{figure}

\section{The Robust-Explanation Landscape}\label{sec:landscape}

\begin{theorem}[Audit complexity]\label{thm:landscape}
For an explicit boundary atlas $(K,D)$, record $t$, and threshold $\theta$:
\begin{center}
\small
\setlength{\arrayrulewidth}{0.4pt}\renewcommand{\arraystretch}{1.25}
\begin{tabular}{@{}l>{\centering\arraybackslash}p{2.05cm}>{\centering\arraybackslash}p{2.05cm}@{}}
\toprule
 & rejected\newline$A(t){=}0$ & accepted\newline$A(t){=}1$ \\
\midrule
\textsc{verify} core & \cellcolor{green!14}\textbf{P} & \cellcolor{orange!22}\textbf{coNP-c} \\
\textsc{min-core} & \cellcolor{orange!30}\textbf{NP-c} & \cellcolor{red!40}\textbf{$\Sigma_2^p$-c}\,$^{\star}$ \\
\bottomrule
\end{tabular}\\[2pt]
{\scriptsize colour $=$ hardness level (light\,$\to$\,dark); $^{\star}$ technical heart.}
\end{center}
\end{theorem}

The four cells are independent results; we sketch each, with the $\Sigma_2^p$-completeness---the technical core---last. Full proofs are in the supplement.

\paragraph{Verify, rejected side (P).}
When $A(t)=0$, $X$ is a robust core iff \emph{every} anchor $k\subseteq t$ is pinned: there is $d\in D$ with $k\subseteq d\subseteq X\cup k$ (so $d\setminus k\subseteq X$). ($\Leftarrow$) any $s\in[X,t]$ containing $k$ contains the pinning $d$, killing the branch, so $A(s)=0$. ($\Rightarrow$) the record $X\cup k$ is rejected, so $k$ is already defeated inside $X\cup k$. The test is an $O(|K||D|n)$ scan.

\paragraph{Min-core, rejected side (NP-complete).}
Membership is clear (guess $X$; verify in P). Hardness is from Vertex Cover: given $G=(V,E)$, take fields $\{a_e\}_{e\in E}\cup\{x_v\}_{v\in V}$, $K=\{\{a_e\}\}$, and $D=\{\{a_e,x_u\},\{a_e,x_v\}\}$ for $e=\{u,v\}$, with $t=H$. Then $A(t)=0$, and by the pinning characterization $X$ is a robust core iff its vertex part is a vertex cover; the anchor fields are never needed. So $\sigma(t)=\tau(G)$.

\paragraph{Verify, accepted side (coNP-complete).}
When $A(t)=1$, $X$ is a robust core iff no record in $[X,t]$ is rejected; a single rejected $s$ refutes it, so the problem is in coNP. Hardness is from \textsc{unsat}: encode a CNF $\varphi$ so that the rejected records in $[X,t]$ are exactly the satisfying assignments of $\varphi$, making $X$ a robust core iff $\varphi$ is unsatisfiable.

\paragraph{Min-core, accepted side ($\Sigma_2^p$-complete).}
Membership: guess $X$ ($\exists$), then ``$X$ is a robust core'' is a coNP check, so the problem is in $\Sigma_2^p$. The lower bound is the heart of the paper.

\begin{theorem}[$\Sigma_2^p$-hardness]\label{thm:sigma2}
\textsc{min-core} on the accepted side is $\Sigma_2^p$-hard.
\end{theorem}

We reduce from a monotone implicant problem.

\begin{definition}[\textsc{mon-dnf-implicant}]
Given a DNF $F(P,Y)$ in which the \emph{selectable} variables $P$ occur only positively (the universal variables $Y$ occur freely) and an integer $\theta$, decide whether some $X\subseteq P$, $|X|\le\theta$, satisfies $F$ under \emph{every} extension---i.e.\ $X$ is a positive implicant of $F$ of size $\le\theta$.
\end{definition}

\begin{lemma}\label{lem:mondnf}
\textsc{mon-dnf-implicant} is $\Sigma_2^p$-complete.
\end{lemma}
Membership is immediate. Hardness is from $\exists\bar x\,\forall\bar y\,\psi$ with $\psi$ a DNF: introduce two positive tokens $p_i^0,p_i^1$ per $x_i$, substitute $x_i\mapsto p_i^1,\neg x_i\mapsto p_i^0$ into $\psi$ to get a $P$-monotone $\psi^+$, and add a mode bit and an address gadget so that $F=(\neg M\wedge\psi^+)\vee\bigvee_i(M\wedge\mathrm{addr}_i\wedge p_i^{0})\vee\bigvee_i(M\wedge\mathrm{addr}_i\wedge p_i^{1})\vee\bigvee_{c\notin[a]}(M\wedge\mathrm{addr}_c)$. With $\theta=a$: at $M{=}1$, address $i$, $F$ reduces to $p_i^0\vee p_i^1$, so any implicant must contain one token per index---forcing exactly one each under the budget $a$, i.e.\ a Boolean $x$; then at $M{=}0$ with the unchosen tokens free-to-false, $F$ becomes $\psi(x,\bar y)$, so the implicant condition is $\forall\bar y\,\psi(x,\bar y)$ (cf.\ shortest implicants \cite{umans2001shortest}).

\paragraph{The reduction.}
Let $F=\bigvee_{r=1}^m T_r$, $T_r=\bigwedge P_r^P\wedge\bigwedge P_r^Y\wedge\bigwedge_{y\in N_r^Y}\neg y$, where $P_r^P\subseteq P$, $P_r^Y\subseteq Y$, and $N_r^Y\subseteq Y$ are the selectable-positive, universal-positive, and universal-negative literals of term $T_r$ (drop contradictory terms; clamp $\theta:=\min(\theta,|P|)$ to keep sizes polynomial). Put $\Theta=\theta+m+1$ and $q=\Theta+1$. The obstacle the construction must defeat: a single \emph{undefeatable} anchor inside $t$ would be a constant-size robust core, collapsing the $\forall$ quantifier. We use four kinds of fields---a marker $e$, a private label $\lambda_r$ per term, a field per $p\in P$, and four kinds of size-$q$ \emph{blocks} (a block $B_j$ per $y_j$, a guard block $G_\mu$ per marker $\mu\in M_0:=\{e,\lambda_1,\dots,\lambda_m\}$, and a rescue block $R$)---with $t=H$ and:
\begin{itemize}\itemsep0.1em
\item \textbf{Anchors}: rescue $\rho=R$; force $f_\mu=\{\mu\}\cup G_\mu$ ($\mu\in M_0$); term $k_r=\{e,\lambda_r\}\cup P_r^P\cup\bigcup_{y_j\in P_r^Y}B_j$.
\item \textbf{Defeaters}: negative-literal $d_{r,j}=k_r\cup B_j$ ($y_j\in N_r^Y$); guard-poison $g_{r,\mu}=k_r\cup G_\mu$ for each $\mu\in M_0$ with $\mu\notin k_r$.
\end{itemize}
The promise $A(t)=1$ holds because $\rho=R$ stands at $t=H$ (no defeater contains $R$). Each gadget plays one role (Table~\ref{tab:gadget}).

\begin{table}[t]
\centering\footnotesize
\begin{tabular}{@{}ll@{}}
\toprule
Gadget & Purpose \\
\midrule
marker $e$, labels $\lambda_r$ & encode the existential layer (a term) \\
block $B_j$ (size $q$) & encode a universal $Y$-assignment \\
guard block $G_\mu$ & stop small cores using force anchors \\
rescue block $R$ & force $A(t){=}1$; too big for a small core \\
term anchor $k_r$ & encode DNF term $T_r$ \\
neg.\ defeater $d_{r,j}$ & kill $k_r$ when $\neg y_j$ is violated \\
guard-poison $g_{r,\mu}$ & kill wrong-marker term anchors \\
\bottomrule
\end{tabular}
\caption{Roles of the $\Sigma_2^p$-hardness gadgets.}
\label{tab:gadget}
\end{table}

\begin{lemma}[Semantics]\label{lem:sem}
For $s\subseteq t$ with $M_0\subseteq s$, $R\not\subseteq s$, and $G_\mu\not\subseteq s$ for all $\mu$, let $\beta_s(p)=[p\in s]$ and $\beta_s(y_j)=[B_j\subseteq s]$. Then $A_{K,D}(s)=1\iff F(\beta_s)=1$.
\end{lemma}
\emph{Proof.} The conditions kill $\rho$ (no $R$) and every $f_\mu$ (no $G_\mu$), leaving only term anchors. $k_r\subseteq s$ iff the positive literals of $T_r$ hold; $k_r$ is defeated by $d_{r,j}\subseteq s$ iff $B_j\subseteq s$ iff the negative literal $\neg y_j$ is violated; guard-poison cannot fire (no full $G_\mu$); and the private labels $\lambda_r$ block any cross-term defeat. So a term anchor stands iff its term is satisfied. \hfill$\square$

\begin{lemma}[Point-restriction]\label{lem:pr}
A robust core of $t=H$ of size $\le\Theta$ exists iff one exists containing only markers and $P$-fields.
\end{lemma}
\noindent
This is what makes the $\forall$ quantifier bite: since $|X|\le\Theta<q$, no robust core completes a block, and a block partially present in $X$ is inert (it is full in $s$ only if $s$ contains all $q$ of its fields, which $t=H$ supplies independently of $X$). Hence the reachable profiles from $[X,t]$ depend only on $X$'s markers and $P$-fields---the existential layer---while the universal layer ranges freely over block completions (full proof in the appendix).

\paragraph{Correctness.}
By Lemma~\ref{lem:pr} it suffices to consider robust cores contained in the markers and the $P$-fields. ($\Rightarrow$) If $X_P\subseteq P$, $|X_P|\le\theta$ is a positive implicant of $F$, take $X=M_0\cup X_P$, $|X|\le\Theta$. For $s\in[X,t]$: if $R\subseteq s$ then $\rho$ stands; else if some $G_\mu\subseteq s$ then $\mu\in M_0\subseteq X$ gives $f_\mu\subseteq s$ standing (no defeater contains $f_\mu$, since guard-poison $g_{r,\mu}$ omits $\mu$); else Lemma~\ref{lem:sem} applies and $X_P\subseteq s$ forces $F(\beta_s)=1$. So $A(s)=1$ throughout and $X$ is a robust core. ($\Leftarrow$) Let $X$ be a robust core over points, $|X|\le\Theta$. First $M_0\subseteq X$: if $\mu\notin X$, then $s=X\cup G_\mu$ has no standing anchor (the force anchor $f_\mu$ lacks $\mu$; every other anchor is either missing $\mu$ or guard-poisoned), so $A(s)=0$, contradicting robustness. Then $X_P=X\cap P$ has $|X_P|\le\theta$, and for any $P$-extension $\alpha$ and any $Y$-assignment $\eta$ the record $s=X\cup\{p:\alpha(p)\}\cup\bigcup_{\eta(y_j)}B_j$ satisfies Lemma~\ref{lem:sem} with $\beta_s=(\alpha,\eta)$; robustness gives $F(\alpha,\eta)=1$, so $X_P$ is a positive implicant. \hfill$\square$

\paragraph{What is new.}
That \emph{some} implicant problem is $\Sigma_2^p$-hard is known (shortest implicants \cite{umans2001shortest}). The contribution is that the \emph{robust-core audit query on a defeasible classifier} is $\Sigma_2^p$-hard: a robust core is not an implicant of a fixed formula but a sufficiency object stable under \emph{all upward disclosures}, and it is the anchor/defeater promise structure---keeping $A(t)=1$ while every further disclosure is quantified over---that realizes the second level. The reduction is \emph{representation-preserving}: the hard instance is a valid defeasible boundary atlas (a legitimate anchored pair meeting the accepted-record promise, with no trivial undefeatable small anchor), not an arbitrary propositional formula, so the hardness survives the entry/exit representation rather than being inherited from shortest-implicants in the abstract. The reduction is verified end to end at the \emph{field level} (real size-$q$ blocks, the robust core searched over \emph{all} fields): across $30$ field-level true-brute instances, $6726$ profile-level multi-term instances, and targeted adversarial cases, the brute-force source answer matches the brute-force atlas robust-core answer with $0$ disagreements, sanity-checking the block barrier the argument relies on.

\begin{proposition}[Small-bad-world tractability]\label{prop:fpt}
For an accepting record $t$, let $B_t$ be the family of $\subseteq$-maximal rejected records $s\subseteq t$ and $m=|B_t|$. A minimum robust core of $t$ is exactly a minimum hitting set of $\{t\setminus s:s\in B_t\}$; grouping fields by their incidence pattern over $B_t$ gives at most $2^m$ field types, so a minimum robust core is computable in $2^{O(m)}\mathrm{poly}(n)$ once $B_t$ is given.
\end{proposition}
\noindent
So $m$, the number of bad worlds, is the parameter that governs accepted-side auditing---the quantity we measure empirically in Section~\ref{sec:exp}. This does not contradict the $\Sigma_2^p$-completeness of Theorem~\ref{thm:sigma2}: the family $B_t$ is \emph{not} part of the atlas input and can itself be exponentially hard to generate; the proposition identifies the parameter measured in our exhaustive audits, not a general FPT algorithm from the raw atlas.

\section{Disclosure and Cheap Queries}

\begin{definition}[\textsc{min-disclosure}]
Given an atlas $(K,D)$, a fixed branch $k\in K$, a finite set of \emph{bad worlds} $W=\{w_1,\dots,w_m\}\subseteq 2^H$, and nonnegative integer field costs $c_h$ (encoded in binary), a disclosure $R\subseteq H$ is \emph{feasible} if it (i) keeps $k$ standing---$k\subseteq R$ and no $d\in D$ has $k\subseteq d\subseteq R$---and (ii) excludes every bad world---$R\not\subseteq w_i$ for all $i$. The \emph{optimization} problem minimizes $\sum_{h\in R}c_h$ over feasible $R$; the \emph{decision} problem asks, given a budget $\beta$, whether a feasible $R$ with $\sum_{h\in R}c_h\le\beta$ exists.
\end{definition}
\begin{theorem}[Minimal certified disclosure]\label{thm:disc}
The decision version of \textsc{min-disclosure} is NP-complete, already with no defeaters (Set Cover); the optimization version is NP-hard. In the defeater-free case the optimization problem is fixed-parameter tractable in the number $m$ of excluded worlds---$2^{O(m)}\mathrm{poly}(n)$, by a subset-DP over the $2^m$ coverage patterns; with general defeaters the natural algorithm is XP ($n^{O(m)}$) and fixed-parameter tractability is open.
\end{theorem}
\noindent
This is a supporting result; the headline bounds (Theorems~\ref{thm:landscape}--\ref{thm:sigma2}) do not depend on it.

\begin{theorem}[Cheap queries]\label{thm:cheap}
On an explicit atlas: prediction $A_{K,D}(t)=\bigvee_{k}[\bigwedge_{h\in k}t_h\wedge\bigwedge_d\neg([k\subseteq d]\wedge\bigwedge_{h\in d}t_h)]$ is a depth-$3$ formula of size $O(|K||D|n)$ (constant depth in the atlas); active reasons are computed by one $\mathrm{poly}(|K|,|D|,n)$ scan and the defeater frontier of a reason by scanning and subset-minimizing the defeaters above it, both with no fixpoint iteration; equivalence of two atlases is in coNP.
\end{theorem}

\noindent
So reading the atlas is cheap, and the landscape's hardness lives entirely in the \emph{optimization} of robust cores and disclosures---not in the representation (Table~\ref{tab:queries}). (We state prediction as a constant-depth \emph{formula in the atlas size}, not as an $\mathrm{AC}^0$ membership claim about $A$: the atlas can be exponentially larger than a succinct description of $A$.)

\begin{table}[t]
\centering\small
\begin{tabular}{@{}lc@{}}
\toprule
Query & Complexity \\
\midrule
Prediction $A(t)$ & depth-3 formula, $O(|K||D|n)$ \\
Active reasons / frontier $D_k$ & poly: scan, then subset-minimize \\
\textsc{verify} core \ (reject / accept) & P\ /\ coNP-c \\
\textsc{min-core} \ (reject / accept) & NP-c\ /\ $\Sigma_2^p$-c \\
Min.\ certified disclosure & dec.\ NP-c; opt.\ FPT (def.-free) \\
Equivalence of two atlases & in coNP \\
\bottomrule
\end{tabular}
\caption{Query support for boundary atlases: reading the atlas is cheap; the hardness is concentrated in robust-core and disclosure \emph{optimization}.}
\label{tab:queries}
\end{table}

\section{Bridge to Formal XAI}

\begin{proposition}[Atlas entries refine \AXp]\label{prop:bridge}
Under the one-sided positive-disclosure encoding, for a monotone $A$ and an accepting record $t$, the subset-minimal entries $k\subseteq t$ coincide with the positive-literal \AXp{}s of $t$, i.e.\ the prime implicants of $A$ contained in $t$. (For general signed encodings, absent literals must be represented as explicit value-literals; the atlas then refines the positive branch reason by adding its defeater frontier.) In general, an entry $k$ standing at $t$ is a branch-sufficient reason ($A\equiv 1$ on $[k,t]$), its minimal extending defeaters form a branch-level contrastive frontier, and a robust core of an accepting $t$ is a minimum-cardinality sufficient reason that survives all further disclosure within $t$. Label-level contrastive explanations are recovered only when the defeated reason has no alternative standing anchor.
\end{proposition}

\noindent
So the hard queries of Sections~3--4 are not abstract combinatorics dressed up as explanation: a robust core \emph{is} a sufficient reason---the familiar \AXp{} in the monotone fragment---made disclosure-aware, and its defeater frontier is the contrastive counterpart. The hardness we prove is the hardness of the explanations practitioners already ask for, once an explanation must survive what is disclosed next.

\paragraph{Audit pipeline.}
End to end: (i) compile a symbolic classifier (tree or rule list) into its lower-cover faces $(\Kstar,\Dstar)$ (Lemma in the appendix); (ii) answer prediction and active reasons by polynomial scans and a reason's defeater frontier by scanning and subset-minimizing the defeaters above it (Theorem~\ref{thm:cheap}); (iii) for an accepting record, enumerate its maximal rejected sub-records on the small exact cube and reduce robust-core minimization to a minimum hitting set over them; (iv) solve that hitting set exactly when the controlling parameters are small---the regime we find in these exact symbolic tabular audits---with Theorem~\ref{thm:sigma2} explaining the worst case.

\section{Experiments}\label{sec:exp}

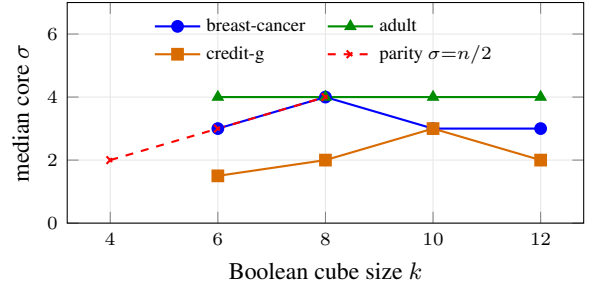
\begin{figure}[t]
\centering
\begin{tikzpicture}
\begin{axis}[
  width=\columnwidth, height=4.5cm,
  xlabel={Boolean cube size $k$}, ylabel={median core $\sigma$},
  xtick={4,6,8,10,12}, ymin=0, ymax=7, ytick={0,2,4,6},
  legend style={font=\scriptsize,at={(0.5,0.98)},anchor=north,draw=none,fill=none,/tikz/every even column/.append style={column sep=6pt}},
  legend columns=2, legend cell align=left,
  label style={font=\footnotesize}, tick label style={font=\scriptsize},
  grid=major, grid style={gray!18},
]
\addplot[mark=*,thick,blue] coordinates {(6,3)(8,4)(10,3)(12,3)}; \addlegendentry{breast-cancer}
\addplot[mark=triangle*,thick,green!55!black] coordinates {(6,4)(8,4)(10,4)(12,4)}; \addlegendentry{adult}
\addplot[mark=square*,thick,orange!85!black] coordinates {(6,1.5)(8,2)(10,3)(12,2)}; \addlegendentry{credit-g}
\addplot[mark=x,thick,red,dashed] coordinates {(4,2)(6,3)(8,4)}; \addlegendentry{parity $\sigma{=}n/2$}
\end{axis}
\end{tikzpicture}
\caption{Robust-core size vs.\ cube size. On real classifiers the median core stays in the low single digits (flat); the adversarial parity family grows linearly ($\sigma{=}n/2$).}
\label{fig:sigma}
\end{figure}

\begin{table*}[t]
\centering
\small
\begin{tabular}{@{}lcccccc@{}}
\toprule
classifier & $n$ & $|K|$ & $|D|$ & $\sigma$ med/max & $\sigma_{\max}/n$ & $m$ max \\
\midrule
breast-cancer & 6 & 16 & 3 & 3/4 & 0.67 & 6 \\
breast-cancer & 8 & 60 & 17 & 4/4 & 0.50 & 8 \\
breast-cancer & 10 & 256 & 164 & 3/5 & 0.50 & 9 \\
breast-cancer & 12 & 1022 & 648 & 3/5 & 0.42 & 10 \\
credit-g & 6 & 17 & 26 & 1.5/4 & 0.67 & 4 \\
credit-g & 8 & 68 & 117 & 2/4 & 0.50 & 5 \\
credit-g & 10 & 256 & 430 & 3/4 & 0.40 & 6 \\
credit-g & 12 & 1132 & 1861 & 2/4 & 0.33 & 8 \\
adult & 6 & 9 & 1 & 4/4 & 0.67 & 5 \\
adult & 8 & 40 & 4 & 4/4 & 0.50 & 5 \\
adult & 10 & 176 & 32 & 4/4 & 0.40 & 7 \\
adult & 12 & 704 & 144 & 4/4 & 0.33 & 7 \\
\midrule
\multicolumn{7}{@{}l}{\emph{adversarial (from the reductions):}}\\
parity chain & 4,6,8 & \multicolumn{5}{l}{robust core $\sigma=n/2$ (linear)} \\
trap Vertex-Cover & 4,6,8 & \multicolumn{5}{l}{greedy/opt $=1.00,1.33,1.50\times$ (grows)} \\
\bottomrule
\end{tabular}
\caption{Over \emph{all} accepting records of each cube, robust cores on real classifiers stay small in \emph{absolute} terms ($\sigma_{\max}/n$ ranges $0.67$ to $0.33$ across the separate per-$k$ cubes, not a nested scaling sequence), with greedy optimal on $98.3$--$100\%$; on adversarial instances built from our reductions the core size and the greedy gap both blow up.}
\label{tab:exp}
\end{table*}

The goal is not to benchmark a scalable solver but to diagnose whether the parameters controlling \emph{exact} auditing are small on standard symbolic/tabular audits, and to confirm that the worst case is nonetheless reachable. We compile decision trees (synthetic non-monotone predicates and three standard tabular datasets under a deliberately small Boolean abstraction: German credit and Adult/income via OpenML \cite{vanschoren2014openml,kohavi1996adult}, Breast Cancer Wisconsin Diagnostic \cite{wolberg1995breast} via scikit-learn \cite{pedregosa2011sklearn}) to exact boundary atlases over $k\le 12$ binarized features, so the robust-core computation is exhaustive.

\paragraph{Worked audit.}
For the eight-field extension of the lending rule list in Example~\ref{ex:credit} (adding \textsf{collateral}, \textsf{long\_history}, \textsf{high\_dti}, \textsf{fraud\_flag}: \textsf{fraud\_flag} is an unconditional defeater, \textsf{high\_dti} is a defeater unless \textsf{collateral} is present, and an approval defeated by \textsf{recent\_default} is reinstated by \textsf{verified\_guarantor}), the compiled atlas has $41$ entries and $54$ defeaters. For an applicant approved via reinstatement, the atlas returns the entry $\{\textsf{income\_stable},\textsf{low\_debt},\textsf{verified\_guarantor}\}$ together with its \emph{global} defeater frontier $\{{+}\textsf{high\_dti}\},\{{+}\textsf{fraud\_flag}\}$, i.e.\ counterfactual additional evidence atoms that would defeat this branch \emph{if admitted into the audit universe}. The robust core reported below is relative to the chosen complete-record universe; when such atoms lie outside that universe they are frontier (contrastive) information, not constraints in the $t$-relative robustness check. A signed \AXp{} can record current absences, but it does not enumerate this defeater frontier as a first-class audit object; the atlas does, and here the robust core coincides with the positive part of the \AXp. All reading queries are microsecond-scale in this Python implementation; exact per-query timings are hardware-dependent and reported only as an artifact sanity check. Only \emph{minimal}-robust-core search is the $\Sigma_2^p$-hard query.

\paragraph{Exact tabular audits: typical cores are small.}
Table~\ref{tab:exp} reports, over \emph{all} accepting records of each cube, the robustness radius $\sigma$ and the FPT parameter $m$ (number of bad worlds) across exact cubes with $k=6,8,10,12$ selected Boolean features. The \emph{median} robust core stays in the low single digits ($2$--$4$ across datasets) and the worst case stays small in absolute terms, with $\sigma_{\max}/n$ ranging $0.67$ down to $0.33$ across the cubes (since each $k$ re-selects and retrains on its own top-$k$ Boolean features, these are separate cubes rather than a single nested scaling sequence, so we read this as ``cores stay small,'' not as a scaling law). The greedy hitting set attains the optimum on $98.3$--$100\%$ of accepting records (only breast-cancer at $k{=}8$ falls below $100\%$). The small typical core is not an artifact of capping tree depth: a control on a fixed feature set finds $\sigma_{\max}$ saturating at $n/2$ already, identical for depth-$6$ and unbounded-depth trees, with median $\sigma$ held at $\approx 2$ throughout. Across $5$ training seeds the radius is stable (median $\sigma=3.2\pm0.4,\,2.5\pm0.5,\,4.0\pm0.0$; max $\sigma=4.8\pm0.4,\,4.0,\,4.0$), and its distribution is concentrated (on the rule list of Example~\ref{ex:credit}, $\sigma$ percentiles $p_{25}/p_{50}/p_{75}/\max=2/2/2/4$). The accept-side verification cost we observe grows as $2^{|t\setminus X|}$ (the natural coNP sweep) while the reject-side pinning verifier stays $O(|K||D|)$, illustrating the coNP-vs-P split of Theorem~\ref{thm:landscape}.

\paragraph{Model quality, and reasons vs robust cores.}
The compiled trees are nontrivial fitted models rather than hand-coded degenerate predicates: at $k{=}12$ they have $33$--$62$ leaves and in-sample accuracy $0.78$/$0.82$/$0.97$ (credit/income/medical), evidence of nondegeneracy rather than a generalization claim. The robust core and a classical positive-literal \AXp{} are related but distinct: comparing the cardinality-minimal robust core with the atlas's cardinality-minimal positive-literal explanation (a standing entry anchor), under deterministic cardinality-first tie-breaking, the two agree as sets on most accepting records and differ on $1.6$/$4.9$/$27.3\%$ at $k{=}12$, with cardinalities almost always equal ($\le 0.7\%$ differ). A robust core is label-sufficient over the whole interval $[X,t]$, not necessarily a single standing branch anchor---different supersets may be certified by different anchors---so it can be smaller than, equal to, or simply set-different from a branch \AXp{} support. The genuinely new audit object is the \emph{global branch frontier} the atlas attaches to each reason (median size $0$--$3$ at $k{=}12$): a bare \AXp{} does not enumerate this branch-level counterfactual frontier, and while contrastive explanations \cite{ignatiev2020contrastive} recover related information, the atlas attaches it directly to each standing reason---exactly what the queries of Section~\ref{sec:landscape} optimize. These figures are regenerated by the artifact (\texttt{model\_quality.py}, \texttt{axp\_robust\_compare.py}).

\paragraph{Adversarial instances bite.}
On constructed instances from our reductions the hardness appears: the parity chain $A(s)=r(s)\bmod 2$ has robust core $\sigma=n/2$ (linear, not $O(1)$), and on the Vertex-Cover atlas the greedy core / optimal core ratio rises with $n$ ($1.00\to1.33\to1.50$ for $n=4,6,8$), so exactness matters off the easy regime.

\noindent
As a correctness sanity check (not a contingent finding), the anchored rule reconstructs each learned tree exactly on its full cube, and the atlas entries match the classical \AXp{}s in the monotone fragment.

\section{Related Work}

\paragraph{Complexity of explanations.}
Explanation queries have a rich complexity theory: across model classes \cite{waldchen2021complexity,barcelo2020interpretability}, for trees specifically \cite{izza2022trees}, with sufficient-reason and shortest-implicant problems reaching $\Sigma_2^p$ \cite{umans2001shortest}, and with knowledge compilation as the route to tractable online queries \cite{audemard2020tractable,darwiche2002kcmap,shih2018symbolic}. All of it studies a \emph{fixed} instance over a \emph{static} space. We turn that axis: the object is a reason's \emph{robustness under future disclosure}, and the question is where its audit queries land in the polynomial hierarchy \cite{stockmeyer1976ph}.

\paragraph{Formal and model-agnostic XAI.}
Abductive/contrastive explanations and their duality \cite{ignatiev2019abduction,ignatiev2020contrastive,reiter1987diagnosis,marquessilva2022trustworthy,marquessilva2023logic} are the objects we audit; model-agnostic methods \cite{ribeiro2016lime,lundberg2017shap,ribeiro2018anchors} give no robustness guarantee. We target exact explanations for symbolic models only.

\paragraph{Argumentation and case-based reasoning.}
Defeasible accept/reject boundaries are the staple of case-based argumentation \cite{cyras2016aacbr,paulinopassos2021monotonicity}, which we treat as one source of what gets compiled---without claiming to discover a canonical boundary. Second-level hardness is itself familiar there: skeptical and credulous acceptance are already $\Sigma_2^p$/$\Pi_2^p$ for several semantics. But ours arises from a different place. The alternation is the \emph{incremental-disclosure} quantifier---a small core that holds against \emph{all} future disclosures---not the extension semantics; the contribution is the complexity of auditing the compiled object under that quantifier.

\section{Limitations}

Several limitations bound the scope. (i) \emph{Compiled-object cost.} The boundary atlas can be exponentially larger than a succinct classifier (decision tree, rule list, circuit); our complexity results are about the compiled object, in the knowledge-compilation tradition of paying an offline compilation cost for online query structure. (ii) \emph{Exact symbolic audits, not a black-box benchmark.} Our experiments are exact audits of depth-limited symbolic classifiers over small Boolean abstractions ($k\le 12$ binarized fields), chosen so that robust-core computation is exhaustive; they diagnose the governing parameters, and are not a scalable black-box XAI benchmark, and the small-parameter regime we observe is contingent (the adversarial families show it can fail). (iii) \emph{One-bit binarization.} Mapping each feature to a single Boolean loses information; signed value-literal encodings can represent $x_i{=}0$ and $x_i{=}1$ separately but grow the atlas and are outside our exact experiments. (iv) \emph{Universe-relative guarantee.} Robustness is relative to the supplied audit universe $t$: a core robust within $t$ may cease to be robust if the universe is enlarged to admit additional defeaters, and a one-sided positive-disclosure encoding cannot certify the \emph{absence} of a defeater outside $t$ (absence certificates require signed value-literals or explicit negative-evidence atoms). (v) \emph{Non-uniqueness.} Minimum robust cores and \AXp{}s need not be unique; all empirical comparisons use deterministic cardinality-first tie-breaking, and the reported percentages should be read under that tie-breaking, not as uniqueness claims. Finally, the $\Sigma_2^p$ result is verified computationally and by independent field-level checks but, like all reductions, rests on the gadget analysis we sketch here and fully prove in the supplement.

\section{Conclusion}

We asked \emph{where} the cost of auditing a defeasible explanation lives, and located it. Reading a compiled boundary atlas---prediction, active reasons, and a reason's defeater frontier---is cheap; \emph{optimizing} the audit is not. Verifying a robust core is coNP-complete and finding a minimal one is $\Sigma_2^p$-complete, completing a four-cell P\,/\,coNP-c\,/\,NP-c\,/\,$\Sigma_2^p$-c landscape and placing disclosure-robust explanation, to our knowledge for the first time, at the second level of the polynomial hierarchy. That barrier is a worst case, not the common one: across the exact small-Boolean audits we run, the governing parameter---the number of maximal bad worlds---stays in the low single digits, so exact robust auditing is tractable in that audited regime. Naming both, the second-level barrier and the small parameter below it, is what turns ``audit a defeasible explanation'' from a slogan into a problem with a precise place in the hierarchy and a usable regime. Natural next steps are tighter atlas-size and uniformity bounds, richer non-subset disclosure orders, and signed encodings that can certify the \emph{absence} of a defeater.

\bibliography{refs}

\appendix
\setcounter{theorem}{0}\setcounter{lemma}{0}\setcounter{definition}{0}\setcounter{proposition}{0}\setcounter{figure}{0}\setcounter{section}{0}
\renewcommand{\thetheorem}{S\arabic{theorem}}
\renewcommand{\thelemma}{S\arabic{lemma}}
\renewcommand{\thedefinition}{S\arabic{definition}}
\renewcommand{\theproposition}{S\arabic{proposition}}
\renewcommand{\thefigure}{S\arabic{figure}}
\renewcommand{\thesection}{\Alph{section}}

This appendix gives the full proofs sketched in the main text (Section/Theorem numbers refer to the main paper), the disclosure and cheap-query results, the bridge, experimental details, and the reproducibility checklist.

\section{Setup recalled}
$H=\{1,\dots,n\}$, record $t\subseteq H$, predicate $A:2^H\to\{0,1\}$ with $A(\emptyset)=0$. An anchored pair $(K,D)$ ($K\cap D=\emptyset$, $\emptyset\notin K$) induces $A_{K,D}(t)=1\iff\exists k\in K:k\subseteq t\wedge\neg\exists d\in D:k\subseteq d\subseteq t$; $k$ is \emph{standing} at $t$ when $k\subseteq t$ and no $d\in D$ has $k\subseteq d\subseteq t$. $A_{K,D}$ is poly-time evaluable. A \emph{robust core} of $t$ is $X\subseteq t$ with $A(s)=A(t)$ for all $s\in[X,t]$; $\sigma(t)=\min|X|$. The \emph{canonical} entries/exits are the lower-cover boundary faces $\Kstar=\{k\neq\emptyset:A(k){=}1,\exists h\in k\,A(k{\setminus}\{h\}){=}0\}$ and $\Dstar=\{d:A(d){=}0,\exists h\in d\,A(d{\setminus}\{h\}){=}1\}$.

\begin{lemma}[Representation]\label{l:repr}
For every finite $A$ with $A(\emptyset)=0$, $A_{\Kstar,\Dstar}=A$. (The faces are \emph{local} lower covers, not globally minimal accepting/rejecting records, so they represent non-monotone $A$ with $1{\to}0{\to}1$ re-entries.)
\end{lemma}
\begin{proof}
Call $u\subseteq t$ \emph{clear} if $A(u)=1$ and no $d\in\Dstar$ has $u\subseteq d\subseteq t$ (i.e.\ $u$ is standing at $t$).

\textbf{($A_{\Kstar,\Dstar}(t){=}1\Rightarrow A(t){=}1$.)} Let $k\in\Kstar$ be standing at $t$, so $A(k)=1$. On any chain $k=u_0\subset\cdots\subset u_\ell=t$, a step with $A(u_i)=1$, $A(u_{i+1})=A(u_i\cup\{h\})=0$ would make $u_{i+1}$ a lower-cover-flip exit, i.e.\ $u_{i+1}\in\Dstar$ with $k\subseteq u_{i+1}\subseteq t$, contradicting standing. So the value never drops along the chain and $A(t)=1$.

\textbf{($A(t){=}1\Rightarrow$ some $k\in\Kstar$ standing.)} The record $t$ is clear (any $d\in\Dstar$ with $t\subseteq d\subseteq t$ is $t$ itself, but $A(t)=1$ so $t\notin\Dstar$), so clear records exist. Suppose for contradiction that no clear record lies in $\Kstar$, and take a $\subseteq$-minimal clear $k$; then $k\notin\Kstar$, so \emph{every} lower cover accepts: $A(k\setminus\{h\})=1$ for all $h\in k$. Fix $h$. By minimality, $k\setminus\{h\}$ is not clear, so some $d\in\Dstar$ has $k\setminus\{h\}\subseteq d\subseteq t$; since $k$ is clear, $k\not\subseteq d$, so $h\notin d$. Now $A(k)=1>0=A(d)$ with $k\subseteq k\cup d\subseteq t$; on a chain from $k$ to $k\cup d$ a $1\to0$ step would give an exit in $\Dstar$ containing $k$ within $t$, contradicting clearness of $k$, so $A(k\cup d)=1$. As $k\cup d=d\cup\{h\}$ has the rejecting lower cover $d$, we get $k\cup d\in\Kstar$; and any exit $d'\in\Dstar$ with $k\cup d\subseteq d'\subseteq t$ would satisfy $k\subseteq d'\subseteq t$, again contradicting clearness of $k$, so $k\cup d$ is clear---a clear $\Kstar$ record, contradiction. Hence some clear $k\in\Kstar$ exists, and it is standing at $t$. The rejected case is dual via $\Dstar$.
\end{proof}

\section{The landscape (Theorem 1)}

\subsection{Verify, rejected side: P}
\begin{lemma}[Pinning]\label{l:pin}
If $A(t)=0$, then $X\subseteq t$ is a robust core iff for every anchor $k\in K$ with $k\subseteq t$ there is $d\in D$ with $k\subseteq d\subseteq X\cup k$ (equivalently $d\setminus k\subseteq X$).
\end{lemma}
\begin{proof}
($\Leftarrow$) Take $s\in[X,t]$ and any $k\subseteq s$. Then $X\cup k\subseteq s$, so the guaranteed $d$ satisfies $k\subseteq d\subseteq s$ and kills the branch; every branch in every such $s$ dies, so $A(s)=0$. ($\Rightarrow$) If $X$ is a robust core and $k\subseteq t$, the record $X\cup k\in[X,t]$ is rejected, so the standing test of $k$ fails inside $X\cup k$: some $d\in D$, $k\subseteq d\subseteq X\cup k$.
\end{proof}
The test is an $O(|K||D|n)$ scan, so \textsc{verify} on the rejected side is in P.

\subsection{Min-core, rejected side: NP-complete}
Membership: guess $X$, verify in P by Lemma~\ref{l:pin}. Hardness from \textsc{vertex cover}.
\begin{proof}
Given $G=(V,E)$, $E=\{e_1,\dots,e_q\}$, build fields $\{a_i\}_{i\le q}\cup\{x_v\}_{v\in V}$, $K=\{\{a_i\}\}$, $D=\{\{a_i,x_u\},\{a_i,x_v\}\}$ for $e_i=\{u,v\}$, and $t=H$. Every anchor $\{a_i\}\subseteq t$ is defeated inside $t$ (e.g.\ by $\{a_i,x_u\}$), so $A(t)=0$. By Lemma~\ref{l:pin}, $X$ is a robust core iff for each $i$, $x_u\in X$ or $x_v\in X$, i.e.\ the vertex part of $X$ is a vertex cover; the anchor fields $a_i$ are never needed (each block $d\setminus\{a_i\}=\{x_\cdot\}$ contains no $a_j$, so dropping $a_i$ preserves the criterion). Hence $\sigma(t)=\tau(G)$.
\end{proof}

\subsection{Verify, accepted side: coNP-complete}
\begin{lemma}\label{l:acc-verify}
If $A(t)=1$, $X$ is a robust core iff no record in $[X,t]$ is rejected. A single rejected $s$ refutes it, so the problem is in coNP.
\end{lemma}
\begin{proof}[Hardness from \textsc{unsat}]
Given CNF $\varphi$ over $z_1,\dots,z_p$ with clauses $c_1,\dots,c_r$, use fields $\{T_i,F_i\}_{i\le p}\cup\{v_i\}_{i\le p}\cup\{c_j\}_{j\le r}\cup\{\varrho\}$. Anchors:
\begin{itemize}\itemsep0.1em
\item rescue $\{\varrho\}$;
\item missing-variable $\{v_i\}$, with defeaters $\{v_i,T_i\}$ and $\{v_i,F_i\}$;
\item inconsistent $\{T_i,F_i\}$ (no defeater, hence undefeatable);
\item unsatisfied-clause $\{c_j\}$, with a defeater $\{c_j,\ell\}$ for each literal field $\ell\in\{T_i,F_i\}$ whose literal satisfies $c_j$.
\end{itemize}
Set $X=\{v_i\}_i\cup\{c_j\}_j$ and $t=H$. The rescue anchor stands at $t$ (no defeater contains $\varrho$), so $A(t)=1$; and $A(\emptyset)=0$. Take any $s\in[X,t]$. If $\varrho\in s$ the rescue stands and $A(s)=1$; otherwise $s=X\cup L$ with $L\subseteq\{T_i,F_i\}_i$, and the standing anchors are exactly: $\{v_i\}$ (iff neither $T_i$ nor $F_i\in L$, i.e.\ variable $i$ unassigned), $\{T_i,F_i\}$ (iff both in $L$, i.e.\ inconsistent), and $\{c_j\}$ (iff no satisfying literal of $c_j$ is in $L$, i.e.\ clause $j$ unsatisfied). Hence $A(s)=0$ iff $L$ assigns every variable, consistently, and satisfies every clause---i.e.\ iff $L$ is a satisfying assignment of $\varphi$. Therefore $X$ is a robust core (no rejected $s\in[X,t]$) iff $\varphi$ is unsatisfiable. Membership is Lemma~\ref{l:acc-verify}.
\end{proof}

\subsection{Min-core, accepted side: $\Sigma_2^p$-complete (Theorem 2)}
Membership: guess $X$ ($\exists$); ``$X$ is a robust core'' is the coNP check of Lemma~\ref{l:acc-verify}; so the problem is in $\Sigma_2^p$.

\paragraph{Source problem.}
\begin{definition}[\textsc{mon-dnf-implicant}]
Given a DNF $F(P,Y)$ with the selectable variables $P$ occurring only positively and an integer $\theta$, decide whether some $X\subseteq P$, $|X|\le\theta$, makes $F$ true under every extension of $X$ over $P\cup Y$.
\end{definition}
\begin{lemma}\label{l:s2c}
\textsc{mon-dnf-implicant} is $\Sigma_2^p$-complete.
\end{lemma}
\begin{proof}
Membership is immediate. For hardness, reduce from $\Phi=\exists x_1\cdots x_a\,\forall y_1\cdots y_b\,\psi$ with $\psi$ a DNF ($\Sigma_2^p$-complete). Introduce positive tokens $p_i^0,p_i^1$ per $x_i$; substitute $x_i\mapsto p_i^1$, $\neg x_i\mapsto p_i^0$ into $\psi$ to obtain $\psi^+$, monotone in $P=\{p_i^b\}$. Add a mode variable $M$ and $\lceil\log a\rceil$ address variables $A$ ($M,A\in Y$), with $\mathrm{addr}_i(A)$ the cube selecting address $i$. Set
\[
\begin{aligned}
F={}&(\neg M\wedge\psi^+)\ \vee\ \textstyle\bigvee_{i=1}^a(M\wedge\mathrm{addr}_i\wedge p_i^0)\\
&\vee\ \textstyle\bigvee_{i=1}^a(M\wedge\mathrm{addr}_i\wedge p_i^1)\ \vee\ \textstyle\bigvee_{c\notin[a]}(M\wedge\mathrm{addr}_c),
\end{aligned}
\]
$\theta=a$.

($\Phi$ true $\Rightarrow$) Let $\xi$ witness $\forall\bar y\,\psi$. Take $X_\xi=\{p_i^{\xi_i}\}$, $|X_\xi|=a$. At $M=1$, any address $i\in[a]$: $F=p_i^0\vee p_i^1$, satisfied since $p_i^{\xi_i}\in X_\xi$; any address $c\notin[a]$: the last disjunct gives $F=1$. At $M=0$: $F=\psi^+$, and since $\psi^+$ is monotone in $P$ the worst extension sets the unchosen tokens false, giving $\psi(\xi,\bar y)=1$ for all $\bar y$. So $X_\xi$ is a positive implicant.

($\exists X,|X|\le a$ implicant $\Rightarrow\Phi$ true) At $M=1$, address $i$, $F=p_i^0\vee p_i^1$, which a free extension can falsify unless $X\cap\{p_i^0,p_i^1\}\neq\emptyset$; over all $i$ that forces $|X|\ge a$, hence exactly one token per index, defining $\xi$. At $M=0$ with the opposite tokens free-to-false, $F=\psi(\xi,\bar y)$; the implicant property gives $\forall\bar y\,\psi(\xi,\bar y)$, so $\Phi$ is true.
\end{proof}

\paragraph{The reduction.}
Let $F=\bigvee_{r=1}^m T_r$, $T_r=\bigwedge_{p\in P_r^P}p\wedge\bigwedge_{y\in P_r^Y}y\wedge\bigwedge_{y\in N_r^Y}\neg y$ (drop terms with $y$ and $\neg y$; clamp $\theta:=\min(\theta,|P|)$, which preserves the answer since $X\subseteq P$ always has $|X|\le|P|$, and keeps $q$ polynomial). Set $\Theta=\theta+m+1$, $q=\Theta+1$. Fields: marker $e$; private label $\lambda_r$ per term; a field per $p\in P$; size-$q$ blocks $B_j$ (per $y_j$), $G_\mu$ (per $\mu\in M_0:=\{e,\lambda_1,\dots,\lambda_m\}$), and $R$. A record $s$ has $y_j(s)=1$ iff $B_j\subseteq s$. Take $t=H$.
\begin{itemize}
\item Anchors: rescue $\rho=R$; force $f_\mu=\{\mu\}\cup G_\mu$ ($\mu\in M_0$); term $k_r=\{e,\lambda_r\}\cup P_r^P\cup\bigcup_{y_j\in P_r^Y}B_j$.
\item Defeaters: negative-literal $d_{r,j}=k_r\cup B_j$ ($y_j\in N_r^Y$); guard-poison $g_{r,\mu}=k_r\cup G_\mu$ for each $\mu\in M_0$ with $\mu\notin k_r$.
\end{itemize}

\emph{Legitimacy.} $d_{r,j},g_{r,\mu}$ strictly contain $k_r$ (using $B_j\not\subseteq k_r$ for $y_j\in N_r^Y$, guaranteed by dropping contradictory terms), and $g_{r,\lambda_\ell}$ ($r\ne\ell$, contains $\lambda_r$, not $\lambda_\ell$) differs from $f_{\lambda_\ell}$ (contains $\lambda_\ell$); so $K\cap D=\emptyset$. \emph{Promise:} $\rho=R$ stands at $t=H$ (no defeater contains $R$), so $A(t)=1$.

\begin{figure}[t]
\centering
\begin{tikzpicture}[font=\footnotesize,>={Latex[length=2mm]},node distance=4.2mm,
  b/.style={draw,rounded corners=2pt,align=center,inner sep=4pt,text width=7.3cm,minimum height=6.5mm},
  src/.style={b,fill=blue!6},
  g/.style={b,fill=gray!8},
  res/.style={b,fill=red!8,draw=red!55},
  fin/.style={b,fill=green!10,draw=green!55!black}]
  \node[src] (s) {\textsc{mon-dnf-implicant} \ (source, $\Sigma_2^p$-complete)};
  \node[g,below=of s]  (ex) {\textbf{markers} $M_0{=}\{e,\lambda_r\}$ \textbf{+ $P$-fields}: existential layer---\\pick one term and its positive tokens};
  \node[g,below=of ex] (un) {\textbf{size-$q$ blocks} $B_j$: a universal $Y$-assignment, realised\\\emph{only} by later disclosure (never fits inside a small core)};
  \node[g,below=of un] (gd) {\textbf{guard blocks} $G_\mu$ \textbf{+ guard-poison} $g_{r,\mu}$:\\deny a small core the force-anchor / wrong-marker shortcut};
  \node[res,below=of gd] (rs) {\textbf{rescue block} $R$: forces the promise $A(t){=}1$,\\ yet $|R|{=}q{>}\Theta$, so a small core cannot use it};
  \node[fin,below=of rs] (o) {robust core of size ${\le}\theta$ \ $\Longleftrightarrow$ \ $F$ has a positive implicant of size ${\le}\theta$};
  \draw[->] (s)--(ex); \draw[->] (ex)--(un); \draw[->] (un)--(gd); \draw[->] (gd)--(rs); \draw[->] (rs)--(o);
\end{tikzpicture}
\caption{Reduction roadmap for the $\Sigma_2^p$-hardness. Each gadget enforces one invariant; together they keep the audited record accepted yet force the universal quantifier to range only over admissible later disclosures, so \emph{minimizing} the robust core decides a $\Sigma_2^p$ implicant question. Field-level definitions are above; the invariants are proved in Lemmas~\ref{l:sem}--\ref{l:wlog}.}
\label{fig:roadmap}
\end{figure}
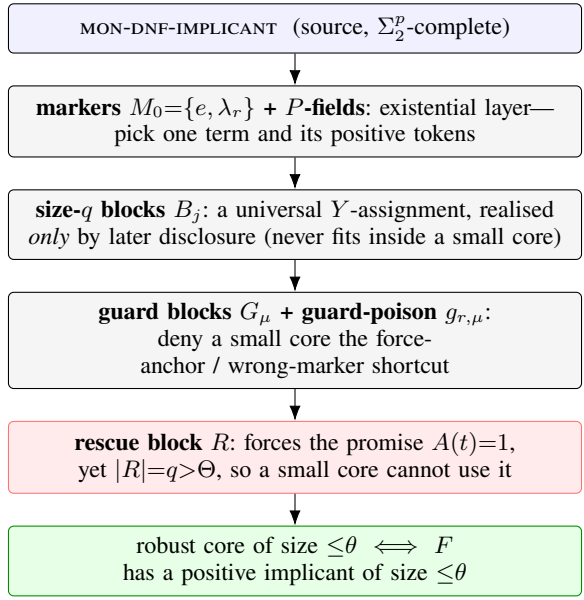

\begin{lemma}[Semantics]\label{l:sem}
For $s\subseteq t$ with $M_0\subseteq s$, $R\not\subseteq s$, and $G_\mu\not\subseteq s$ for all $\mu$, set $\beta_s(p)=[p\in s]$ and $\beta_s(y_j)=[B_j\subseteq s]$. Then $A_{K,D}(s)=1\iff F(\beta_s)=1$.
\end{lemma}
\begin{proof}
$\rho$ is not standing ($R\not\subseteq s$) and no $f_\mu$ is ($G_\mu\not\subseteq s$), leaving only term anchors. $k_r\subseteq s$ iff $\{e,\lambda_r\}\subseteq s$ (true, $M_0\subseteq s$) and $P_r^P\subseteq s$ and $B_j\subseteq s$ for $y_j\in P_r^Y$---i.e.\ all positive literals of $T_r$ hold under $\beta_s$. $k_r$ is defeated by $d_{r,j}\subseteq s$ iff $B_j\subseteq s$ iff $\beta_s(y_j)=1$, i.e.\ the negative literal $\neg y_j$ is violated. Guard-poison $g_{r,\mu}$ cannot fire ($G_\mu\not\subseteq s$). Cross-term defeat is impossible: for $r'\ne r$, every defeater generated from term $r'$ is built from $k_{r'}$ (and possibly a block), hence contains $\lambda_{r'}$ but \emph{not} the private label $\lambda_r$; since $k_r$ contains $\lambda_r$, no such defeater can contain $k_r$. Hence a term anchor stands iff its term is satisfied, and $A_{K,D}(s)=1\iff F(\beta_s)=1$.
\end{proof}

\paragraph{Profile projection.} Partition the fields into \emph{points} $\Pi=M_0\cup P$ and \emph{blocks} (the size-$q$ sets $B_j$, $G_\mu$, $R$, pairwise disjoint). For a record $s\subseteq t$ define its \emph{profile}
\[
\pi(s)=\big(\,s\cap\Pi,\ \{\,\text{block }B:B\subseteq s\,\}\,\big),
\]
the points present together with the blocks \emph{fully} present.

\begin{lemma}[Profile determines the value]\label{l:profile}
$A_{K,D}(s)$ is a function of $\pi(s)$ alone.
\end{lemma}
\begin{proof}
Every anchor and defeater is a union of points and whole blocks: $k_r=\{e,\lambda_r\}\cup P_r^P\cup\bigcup_{y_j\in P_r^Y}B_j$, $f_\mu=\{\mu\}\cup G_\mu$, $\rho=R$, $d_{r,j}=k_r\cup B_j$, $g_{r,\mu}=k_r\cup G_\mu$. Hence each test ``$k\subseteq s$'' and ``$d\subseteq s$'' depends only on which points lie in $s$ and which blocks are full in $s$, i.e.\ on $\pi(s)$; so does the standing condition of every anchor, and therefore $A_{K,D}(s)$.
\end{proof}

\begin{lemma}[Point-restriction is WLOG]\label{l:wlog}
Let $|X|\le\Theta$ and let $X^\circ=X\cap\Pi$ be its point-restriction. Then $X$ is a robust core of $t=H$ iff $X^\circ$ is. Consequently a robust core of size $\le\Theta$ exists iff one contained in the points does.
\end{lemma}
\begin{proof}
Since $|X|\le\Theta<q$, $X$ contains no full block ($X\cap B\subsetneq B$ for every block $B$). As $X^\circ\subseteq X$ we have $[X,t]\subseteq[X^\circ,t]$, so $A\equiv1$ on $[X^\circ,t]$ implies $A\equiv1$ on $[X,t]$: this is the ``$X^\circ$ robust $\Rightarrow X$ robust'' direction.

Conversely let $X$ be robust and take any $s\in[X^\circ,t]$; we exhibit $s^\ast\in[X,t]$ with $\pi(s^\ast)=\pi(s)$, whence $A(s)=A(s^\ast)$ by Lemma~\ref{l:profile}, and $A(s^\ast)=1$ since $s^\ast\in[X,t]$ and $X$ is robust. Put
\[
s^\ast=(s\cap\Pi)\ \cup\ \textstyle\bigcup_{B\subseteq s}B\ \cup\ (X\setminus\Pi),
\]
i.e.\ the points of $s$, the blocks \emph{full} in $s$, and the (partial) block fields of $X$. Then $s^\ast\subseteq t=H$; and $X\subseteq s^\ast$, since $X\cap\Pi=X^\circ\subseteq s\cap\Pi$ (as $X^\circ\subseteq s$) and $X\setminus\Pi\subseteq s^\ast$. So $s^\ast\in[X,t]$. Finally $\pi(s^\ast)=\pi(s)$: the points of $s^\ast$ are exactly $s\cap\Pi$; a block $B$ with $B\subseteq s$ is included whole, while a block $B$ with $B\not\subseteq s$ has $s^\ast\cap B=(X\setminus\Pi)\cap B\subseteq X\cap B\subsetneq B$, so it is \emph{not} full in $s^\ast$ (the dropped partial fields of $s$ and the added partial fields of $X$ complete no block). Hence the full blocks of $s^\ast$ are exactly those of $s$, $\pi(s^\ast)=\pi(s)$, and $A(s)=A(s^\ast)=1$. As $s\in[X^\circ,t]$ was arbitrary, $X^\circ$ is robust.
\end{proof}

\begin{proof}[Correctness of the reduction]
By Lemma~\ref{l:wlog} we may search robust cores over points.

($\Rightarrow$) Let $X_P\subseteq P$, $|X_P|\le\theta$ be a positive implicant of $F$; put $X=M_0\cup X_P$, $|X|\le\Theta$. For $s\in[X,t]$: if $R\subseteq s$, $\rho$ stands; else if some $G_\mu\subseteq s$, then $\mu\in M_0\subseteq X\subseteq s$ gives $f_\mu\subseteq s$, and $f_\mu$ is undefeated (the only defeaters with a full $G_\mu$ are the $g_{r,\mu}$, which omit $\mu$ since they are added only for $\mu\notin k_r$, hence never contain $f_\mu$); else Lemma~\ref{l:sem} applies and $X_P\subseteq s$ forces $\beta_s$ to extend $X_P$, so $F(\beta_s)=1$ and $A(s)=1$. So $A\equiv1$ on $[X,t]$ and $X$ is a robust core of size $\le\Theta$.

($\Leftarrow$) Let $X$ be a robust core, $|X|\le\Theta$; by Lemma~\ref{l:wlog} take $X$ over points. First, $M_0\subseteq X$: if $\mu\notin X$, consider $s=X\cup G_\mu\in[X,t]$. Then $R\not\subseteq s$; the force anchor $f_\mu$ lacks $\mu$; every other $f_{\mu'}$ lacks its guard; $\rho$ lacks $R$; and any term anchor $k_r\subseteq s$ either misses $\mu$ (if $\mu\in k_r$) or is guard-poisoned by $g_{r,\mu}\subseteq s$ (if $\mu\notin k_r$, since then $g_{r,\mu}$ exists and $G_\mu\subseteq s$). So no anchor stands and $A(s)=0$, contradicting robustness. Hence $M_0\subseteq X$, and $X_P:=X\cap P$ has $|X_P|\le\Theta-(m+1)=\theta$. Now for any extension $\alpha$ of $X_P$ over $P$ and any $\eta$ over $Y$, put $s=X\cup\{p:\alpha(p)\}\cup\bigcup_{\eta(y_j)}B_j$. Then $R\not\subseteq s$, $G_\mu\not\subseteq s$ (blocks not completed), $M_0\subseteq s$, so Lemma~\ref{l:sem} gives $A(s)=F(\beta_s)$ with $\beta_s=(\alpha,\eta)$. Robustness forces $A(s)=1$, so $F(\alpha,\eta)=1$ for all $\alpha,\eta$, i.e.\ $X_P$ is a positive implicant of $F$ of size $\le\theta$.
\end{proof}

The reduction is polynomial and preserves the promise $A(t)=1$; with Lemma~\ref{l:s2c} and membership, \textsc{min-core} on the accepted side is $\Sigma_2^p$-complete.

\paragraph{Machine verification.}
We checked the reduction end to end at the \emph{field level} with real size-$q$ blocks and the robust core searched over \emph{all} fields: $30$ field-level true-brute instances, $6726$ profile-level exhaustive multi-term instances (exercising guard-poison and cross-defeat), and targeted adversarial cases. In every instance the brute-force source answer matches the brute-force atlas robust-core answer ($0$ disagreements), an independent finite sanity check of Lemmas~\ref{l:sem}--\ref{l:wlog} and the promise.

\section{Disclosure (Theorem 3)}
\textsc{min-disclosure}: given an atlas, a fixed branch $k$, nonnegative integer costs on fields (encoded in binary), and a set $W$ of $m$ ``bad worlds'' to exclude, find the minimum-cost disclosure $R$ keeping $k$ standing while excluding all of $W$. With no defeaters the decision version (is there a feasible $R$ of cost $\le\beta$?) is exactly weighted Set Cover, hence NP-complete; the optimization version is therefore NP-hard.
\begin{proposition}
In the defeater-free case, \textsc{min-disclosure} is FPT in $m=|W|$, solvable in $2^{O(m)}\mathrm{poly}(n)$.
\end{proposition}
\begin{proof}
Keeping the branch $k$ standing forces $k\subseteq R$, so the disclosure is $R=k\cup R'$ with $R'\subseteq H\setminus k$ free (in the defeater-free case no $d$ can defeat $k$, so any such $R$ keeps $k$ standing). The mandatory part already excludes the worlds $\mathrm{cov}(k)=\{w\in W:k\not\subseteq w\}$ at fixed cost $\mathrm{cost}(k)$. It remains to cover $W_0:=W\setminus\mathrm{cov}(k)$ with fields from $H\setminus k$. Each field $h$ excludes a fixed $\mathrm{cov}(h)\subseteq W$, computable in $\mathrm{poly}(n)$; the objective depends on a chosen field set only through the union of its $\mathrm{cov}(h)$, of which there are $\le 2^{|W_0|}\le 2^m$ distinct values. Group the fields of $H\setminus k$ by $\mathrm{cov}(h)\cap W_0$, keep the cheapest per group, and run the weighted set-cover DP over subsets $S\subseteq W_0$ \emph{seeded by the mandatory cost}: $\mathrm{DP}[\emptyset]=\mathrm{cost}(k)$ and $\mathrm{DP}[S]=\min_{h\in H\setminus k}\big(\mathrm{DP}[S\setminus\mathrm{cov}(h)]+\mathrm{cost}(h)\big)$; the answer is $\mathrm{DP}[W_0]$. This is $2^{O(m)}\mathrm{poly}(n)$.
\end{proof}
With general defeaters the disclosure must also avoid completing any defeater, so coverage is no longer additive over fields; the natural branch-on-worlds algorithm is then $n^{O(m)}$ (XP), and we leave fixed-parameter tractability of the general case open. Theorem~3 is a supporting result; the headline lower bounds (Theorems 1--2) do not depend on it.

\section{Cheap queries (Theorem 4)}
$A_{K,D}(t)=\bigvee_{k\in K}\big[\bigwedge_{h\in k}t_h\wedge\bigwedge_{d\in D}\neg([k\subseteq d]\wedge\bigwedge_{h\in d}t_h)\big]$ with $[k\subseteq d]$ a compile-time constant; this is an unbounded-fan-in depth-$3$ formula of size $O(|K||D|n)$ (no fixpoint). Active reasons $\{k:k\text{ standing at }t\}$: one scan, $O(|K||D|n)$ (precompute $D{\downarrow}t$, test each $k$). Defeater frontier of $k$: scan $\{d\in D:k\subseteq d\}$ and $\subseteq$-minimize, $O(|D|^2n)$ (linear after indexing). Equivalence: a witness of inequivalence is one record, checkable in poly time, so non-equivalence is in NP and equivalence in coNP.

\section{Bridge (Proposition 2)}
If $A$ is monotone then $\Dstar=\emptyset$ (a $d$ with $A(d)=0<A(d\setminus\{h\})$ contradicts monotonicity), the standing condition is vacuous, and $\min_\subseteq\Kstar$ are the $\subseteq$-minimal accepting records, i.e.\ the prime implicants of $A$, i.e.\ the classical \textsc{AXp}{}s. In general, an entry $k$ standing at $t$ satisfies $A\equiv1$ on $[k,t]$ (any defeater $d$ with $k\subseteq d\subseteq s\subseteq t$ would already defeat $k$ at $t$); its minimal extending defeaters form a branch-level contrastive frontier; and a robust core of an accepting $t$ is a minimal sufficient reason stable under all disclosure within $t$. Label-level contrastive explanations coincide with the branch frontier only when the defeated reason has no alternative standing anchor (otherwise the prediction is reinstated).

\section{Experimental details and reproducibility}
\paragraph{Setup.} Decision trees (\textsc{DecisionTreeClassifier}, Gini, \texttt{max\_depth}$=6$) are trained on each original feature binarized to one Boolean (numeric $\ge$ median; categorical $=$ mode); the top-$k$ Booleans by importance are kept and the tree retrained on them, with the accept class oriented so $A(\emptyset)=0$; hence ``accepting'' denotes the oriented target label in the formal audit, not necessarily the semantically favorable class of the original dataset. For Adult we use the first $8{,}000$ OpenML rows before binarization, only to keep the artifact lightweight; the audit itself is over the induced $k$-bit classifier cube, not the raw rows. Note each $k$ re-selects its own top-$k$ Booleans and retrains, so the per-$k$ cubes are separate (not a nested feature sequence). The atlas is compiled exactly by lower-cover analysis over the $2^k$ cube ($k\le12$). Robust cores are computed exactly as minimum hitting sets of the maximal-bad-world complements; ties are broken by cardinality first, then a fixed field order. Online-latency reports use repeated queries ($2000$ for the synthetic compiler demo; $150$ atlas / $120$ baseline per real dataset); all timings are hardware-dependent and only illustrative---the reproducible artifact is the set of correctness assertions and the structural CSV fields, which are deterministic given the fixed seeds.
\paragraph{Variance.} Over $5$ training seeds ($k=10$): median $\sigma=3.2\pm0.4$ (breast-cancer), $2.5\pm0.5$ (credit-g), $4.0\pm0.0$ (adult); max $\sigma=4.8\pm0.4,\,4.0\pm0.0,\,4.0\pm0.0$. A depth control (fixed feature set, \texttt{max\_depth}$\in\{3,4,6,8,12,\infty\}$) finds $\sigma_{\max}$ saturating at $n/2$ already at depth $6$ (identical for unbounded depth), so the small typical core is not an artifact of capping depth.
\paragraph{Adversarial families.} Parity chain $A(s)=r(s)\bmod2$ ($r=$ longest present prefix): robust core $\sigma=n/2$. Vertex-Cover trap (the reduction above on the Johnson greedy-bad bipartite graph): greedy/optimal robust-core ratio $1.00,1.33,1.50$ for $n=4,6,8$.

\paragraph{Reproducibility checklist.}
\begin{itemize}\itemsep0.1em
\item All claims are stated with explicit assumptions (finite $H$; explicit atlas input). Yes.
\item Complete proofs of all theoretical results: in this appendix. Yes.
\item Code, data preparation, and scripts to reproduce every table/number: included as supplementary code; deterministic given fixed seeds. Yes.
\item Datasets: German credit (OpenML id 31), Adult (id 1590), Breast Cancer (scikit-learn builtin); public. Yes.
\item Hyperparameters: none are tuned for predictive performance. \texttt{max\_depth}$=6$ is fixed a priori for interpretability and exact-cube feasibility (with a depth-control grid $\{3,4,6,8,12,\infty\}$ reported); $k\in\{6,8,10,12\}$, each $k$ selecting its own top-$k$ Boolean features by tree importance and retraining. Yes.
\item Compute (reference environment): single CPU (AMD Ryzen~9 9950X, 16~cores/32~threads; one core used), no GPU, $\le 2$\,GB RAM, on Windows~11 (also tested on Linux x86\_64); Python 3.13.5, numpy 2.3.5, scikit-learn 1.8.0, pandas 2.2.3. The whole artifact runs in about $1$--$3$ minutes; per-query microsecond timings are hardware-dependent. Yes.
\item Error bars / variance: reported over $5$ seeds above. Yes.
\end{itemize}

\end{document}